\documentclass[11pt]{article}

\usepackage[letterpaper,margin=1in]{geometry}
\usepackage[T1]{fontenc}
\usepackage{amsmath,amssymb,amsthm,mathtools}
\usepackage{newtxtext,newtxmath}
\usepackage{microtype}
\usepackage[hidelinks]{hyperref}
\usepackage[nameinlink,noabbrev]{cleveref}
\usepackage{enumitem}

\setlist{nosep,leftmargin=*}
\allowdisplaybreaks

\newtheorem{theorem}{Theorem}[section]
\newtheorem{lemma}[theorem]{Lemma}
\newtheorem{corollary}[theorem]{Corollary}

\newcommand{\IC}{\mathrm{IC}}
\newcommand{\GapMAJ}{\mathsf{GapMAJ}}
\newcommand{\E}{\mathbb E}
\newcommand{\Var}{\operatorname{Var}}

\title{Constant-Loss Round-Preserving Direct Sums for Information Complexity}
\author{Haoran Wang\\\small\texttt{whr.hrwang@gmail.com}}
\date{}

\begin{document}
\maketitle

\begin{abstract}
Strong composition theorems characterize when solving a composed problem requires essentially the sum of the information costs of its components. Sawettamalya and Yu proved an optimal Gap-Majority composition theorem, and showed that preserving the number of communication rounds incurs an additional logarithmic loss. We prove a round-preserving variance direct sum with constant loss for every fixed finite terminal-answer alphabet. The result applies to independent square-integrable observables and yields the constant-loss Gap-Majority lemma without increasing the number of messages.

The proof encodes centered counterfactual conditional means. Their total squared magnitude is controlled by a telescoping variance identity, allowing entropy bounds for the auxiliary messages that do not depend on the length of the sum. We also show that a terminal channel of likelihood rank at most $s$ can be replaced, without changing the transcript or increasing observer variance, by one with at most $s$ output symbols. This gives a rank-only direct sum for large and continuous terminal-answer spaces. For unrestricted outputs, we obtain a quantitative mean-square extraction theorem under a fourth-moment assumption.
\end{abstract}

\section{Introduction}
For a Boolean function $f:\mathcal X\times\mathcal Y\to\{\pm1\}$, the gap-majority composition asks Alice and Bob to distinguish $\sum_{i=1}^n f(X_i,Y_i)\ge 0.01\sqrt n$ from $\sum_{i=1}^n f(X_i,Y_i)\le-0.01\sqrt n$. Since the inputs are independent, one can solve the problem by estimating each instance with sufficiently small error, calibrating any systematic bias, and summing the estimates. For instance, an unbiased estimate with sufficiently small constant variance makes the sum accurate on the $\sqrt n$ scale. A strong composition theorem asks whether the converse holds: does a successful protocol for the sum reveal information proportional to the number of instances?

Information cost makes this question precise. For a protocol with transcript $M$ and public randomness $R$, it measures the information the transcript reveals about each player's input beyond the other player's input. Unlike communication length, it has an exact chain rule, which makes it well suited to direct-sum arguments~\cite{BBCR10,BR11,Bra15}. Preserving the number of messages is equally important in one-way and bounded-round communication. A reduction that adds interaction need not preserve the relevant lower bound.

\subsection{Prior work and the round constraint}
Strong composition theorems in two-party communication were first developed for the identity and XOR outer functions. Round-preserving compression and bounded-round XOR lemmas provided tools for retaining the interaction structure~\cite{BRWY13,Yu22,IR24,SY25}. Sawettamalya and Yu~\cite{SY26} proved a strong Gap-Majority lemma with information cost $\mathrm{IC}(\pi)/n+O(1)$. Their reduction does not preserve the original number of messages. They also gave a round-preserving reduction with additive $O(\log n)$ information loss. The question left open was whether constant loss and round preservation can hold simultaneously.

The extra logarithm arises during coordinate decomposition, rather than from the information chain rule. Alice must prepare the conditional means required for each possible final answer, even though Bob chooses the answer only after her last message. Encoding these means at a uniform precision has a cost that grows with their range. The approach here is to encode only the part not determined by the common view. Its range may still grow with $n$, but its total squared magnitude admits a global bound.

\subsection{Main results}

We work with independent input pairs $(X_i,Y_i)\sim\mu_i$ and square-integrable real observables $g_i(X_i,Y_i)$. Let $\pi$ be an $r$-message protocol in which Alice sends the last message and Bob produces a terminal answer $A$. Write $M$ for its transcript and $R$ for its public randomness. The internal information cost is
\[
\IC(\pi)=I(M:X\mid Y,R)+I(M:Y\mid X,R).
\]
For a square-integrable random variable $Z$, following~\cite{SY26} we write
\[
\Phi_{\rm obs}(\pi;Z)=\E\Var(Z\mid M,R,A),
\qquad
\Phi_{\rm B}(\pi;Z)=\E\Var(Z\mid M,R,A,Y).
\]
The two quantities correspond to the information available to an outside observer and to Bob, respectively.

Our main theorem is stated for heterogeneous product distributions. Assume that, conditioned on every transcript and public random string, the support of Bob's terminal answer has size at most $m$, with a measurable enumeration of that support. Put
\[
\sigma^2=\frac1n\sum_{i=1}^n \Var(g_i(X_i,Y_i)).
\]

\begin{theorem}\label{thm:main}
For every $0<\Delta\le \sigma^2$, there are $r$-round single-coordinate protocols $\eta_1,\ldots,\eta_n$ such that
\[
\frac1n\sum_{i=1}^n \IC_{\mu_i}(\eta_i)
\le
\frac{\IC(\pi)}{n}
+C m\log_2\!\left(2+\frac{\sigma^2}{\Delta}\right)
\]
and
\[
\frac1n\sum_{i=1}^n \Phi_{\rm B}(\eta_i;g_i)
\le
\frac1n\Phi_{\rm obs}\!\left(\pi;\sum_{i=1}^n g_i(X_i,Y_i)\right)+\Delta,
\]
where $C$ is an absolute constant.
\end{theorem}

The theorem also covers weighted sums by incorporating the weights into $g_i$. For i.i.d. inputs, a uniformly chosen coordinate gives a single protocol with the same bounds.

For Gap-Majority, the terminal answer is binary and $\sigma^2=1$. Sawettamalya and Yu prove that a protocol which computes $\GapMAJ_n\circ f^n$ with success probability $0.99$ has observer variance at most $0.99n$~\cite{SY26}. Taking a fixed constant $\Delta$ in \cref{thm:main} therefore gives the conjectured same-round constant-loss extraction.

\begin{corollary}\label{cor:gapmaj}
Let $f:\mathcal X\times\mathcal Y\to\{\pm1\}$ satisfy $\E_\mu f=0$. If an $r$-round protocol $\pi$ computes $\GapMAJ_n\circ f^n$ with probability at least $0.99$ under $\mu^n$, then there is an $r$-round protocol $\eta$ for one copy of $f$ such that
\[
\Pr[\eta=f]\ge0.501
\qquad\text{and}\qquad
\IC_\mu(\eta)\le \frac{\IC_{\mu^n}(\pi)}{n}+O(1).
\]
\end{corollary}

For an arbitrary terminal output, the relevant complexity parameter may be its likelihood rank rather than its cardinality. Section~\ref{sec:rank} gives a terminal-output reduction that replaces a rank-$s$ channel by an at-most-$s$-symbol channel while preserving the variance relevant to the theorem. Section~\ref{sec:mmse} treats terminal channels of unrestricted rank by a quantitative approximation.

\subsection{Proof overview}
Fix a stage of the coordinate decomposition with $k$ coordinates remaining. For a possible terminal answer $a$, let $Z_{k,a}$ be the conditional mean of the residual sum, given Alice's view and the event that Bob outputs $a$. Alice encodes the difference $D_{k,a}$ between this mean and its prediction from the common view. A small posterior resampling of Bob's answer ensures that the squared innovations of all hypothetical branches can be charged to the branch that is actually realized.

Write $G_k=\E D_{k,B}^2$ for the realized squared innovation, $V_k$ for the current observer variance, and $v_k$ for the extracted coordinate's Bob variance. Conditional orthogonality gives
\[
G_k\le 2(V_k-v_k-V_{k-1}+\delta)+2\Var(g_k).
\]
After summing over $k$, the variance terms telescope. An entropy bound for quantized real variables then converts the resulting total $L_2$ estimate into an information-cost bound. The underlying coordinate decomposition is that of~\cite{SY26}; the new step is the centered message and its variance-based cost analysis. For low-rank terminal channels, a separate convex-geometric reduction compresses the original answer before peeling. It requires no recursive rank condition or quantitative bound on the likelihood ratios.

\section{Preliminaries}\label{sec:prelim}

We use the distributional two-party model. Alice receives $X$, Bob receives $Y$, the players share public randomness $R$, and they exchange an $r$-message transcript $M=(M_1,\ldots,M_r)$. Bob produces the terminal answer after the communication. Following~\cite{SY26}, the last of the $r$ messages is sent by Alice, and Bob produces the answer without further communication; in this convention $r$ is odd. The observer variance uses the complete transcript. Bob may use private randomness to choose the terminal answer. We take the input and answer spaces to be standard Borel, so the conditional laws used below are well defined. All auxiliary messages are appended to Alice's last message before Bob produces his terminal answer.

For product inputs we use the standard rectangle property in the form used in bounded-round direct-sum arguments~\cite{Yu22,SY26}. Conditioning on the transcript, public randomness, and disjoint fixed Alice and Bob coordinates preserves the required factorization of the remaining coordinates. A local postprocessing of Bob's terminal answer preserves this property because it only modifies the Bob-side factor.

We will also use the law of total variance. If $\mathcal F\subseteq\mathcal G$, then
\[
\E\Var(Z\mid\mathcal F)
=
\E\Var(Z\mid\mathcal G)
+
\E\bigl(\E[Z\mid\mathcal G]-\E[Z\mid\mathcal F]\bigr)^2.
\]

The only coding estimate needed in the proof is the following elementary entropy bound.

\begin{lemma}\label{lem:entropy}
Let $D$ be square integrable, let $h>0$ be measurable with respect to a sigma-field $\mathcal F$, and put $K=\lfloor D/h\rceil$. Then
\[
H(K\mid\mathcal F)
\le
c_0+\frac12\log_2\!\left(1+c_1\E\frac{D^2}{h^2}\right)
\]
for absolute constants $c_0,c_1$.
\end{lemma}

\begin{proof}
Condition on $\mathcal F$ and add an independent uniform random variable $U$ on $[-1/2,1/2]$. The unit intervals around distinct integers are disjoint, so $H(K\mid\mathcal F)=h(K+U\mid\mathcal F)$. Gaussian maximum entropy bounds the latter by the logarithm of its conditional variance. Since $|K|\le |D|/h+1/2$, Jensen's inequality gives \cref{lem:entropy}.
\end{proof}

\section{Proof of the direct-sum theorem}\label{sec:proof}

We center each observable and write
\[
S_k=\sum_{i=1}^k\bigl(g_i(X_i,Y_i)-\E g_i(X_i,Y_i)\bigr).
\]
This does not change any conditional variance. The proof follows the round-preserving peeling of~\cite{SY26}, but the auxiliary message at each stage is replaced by a quantized innovation. The standard one-step information identities are recorded in the appendix.

\subsection{Counterfactual innovations}

Let $A$ be Bob's original terminal answer and let $q_a(M,R)=\Pr[A=a\mid M,R]$. Fix $0<\rho<1$. After all messages have been determined, Bob forms a new answer $B$ as follows: with probability $1-\rho$ he keeps $A$, while with probability $\rho$ he uses private randomness to resample an independent value from the distribution $q(\cdot\mid M,R)$. Thus
\[
\Pr[B=a\mid M,R]=q_a(M,R).
\]
Moreover, if $\mathcal G$ contains $(M,R)$ and is generated before this final resampling, then
\begin{equation}
\Pr[B=a\mid\mathcal G]
=(1-\rho)\Pr[A=a\mid\mathcal G]+\rho q_a(M,R)
\ge \rho q_a(M,R).
\label{eq:branch-domination}
\end{equation}
This randomization leaves the transcript unchanged. For every square-integrable $Z$, total variance also gives
\begin{equation}
\E\Var(Z\mid M,R,B)
\le
\E\Var(Z\mid M,R,A)+(2\rho-\rho^2)\Var(Z).
\label{eq:smoothing-variance}
\end{equation}
Indeed, conditioned on $(M,R)$, the conditional mean indexed by $B$ is the factor $(1-\rho)$ contraction of the conditional mean indexed by $A$ around their common average.

\begin{lemma}[Pre-answer domination]\label{lem:invariant}
In the successive round-preserving decompositions, the answer $B$ is generated only after all the auxiliary messages of the current protocol have been fixed. At every stage, the simulated original transcript and public randomness $(M,R)$ remain part of the common view, and the sampling law $q_a(M,R)$ is never updated after auxiliary messages are added. For any sigma-field $\mathcal H$ that contains $(M,R)$ and is generated by the inputs, public and private samples, and auxiliary messages before $B$ is drawn, and for any $\mathcal H$-measurable family $d_a$,
\[
\E d_B^2\ge\rho\E\sum_a q_a(M,R)d_a^2.
\]
In particular, this holds at every peeling stage with $d_a=D_{k,a}$.
\end{lemma}
\begin{proof}
The original transcript and its public random string are simulated without change in both children. The current auxiliary messages are functions of Alice's view and the public samples; they are appended to her final message, with one message for every hypothetical value of the answer. They therefore depend neither on the realized $B$ nor on the coin used to resample it. The residual child uses precisely the same terminal-answer rule, including the same fixed map $(M,R)\mapsto q(\cdot\mid M,R)$, and the coordinate child uses its Bob-side version. These properties hold initially and are preserved when another auxiliary vector is appended, establishing the pre-answer invariant by induction. Conditional on $\mathcal H$, the resampling coin is independent and gives $\Pr[B=a\mid\mathcal H]\ge\rho q_a(M,R)$. Multiplication by $d_a^2$ and summation prove the claim.
\end{proof}

Apply this randomization once and keep the same $B$ through the entire coordinate decomposition. At the stage with $k$ live coordinates, let $T_k$ denote the current transcript, including the auxiliary data sent at earlier peeling stages, and let $R_k$ denote the public randomness of the parent protocol. In the residual child, $Y_k$ is newly sampled publicly and $X_k$ privately on Alice's side, as in~\cite{SY26}. For the split we put $F_k=(T_k,R_k,Y_k)$; the newly exposed $Y_k$ is not included in the parent observer variance $V_k$. For every posterior-supported value $a$ of $B$, define
\[
Z_{k,a}=\E[S_{k-1}\mid F_k,X_{<k},B=a],
\qquad
b_{k,a}=\E[S_{k-1}\mid F_k,B=a],
\]
and set $D_{k,a}=Z_{k,a}-b_{k,a}$. Alice can compute $Z_{k,a}$ from her current view for every hypothetical branch, while $b_{k,a}$ is determined by the common view and the specified answer value. The quantity $D_{k,a}$ is the branchwise information that is new to Bob.

Choose a branch-dependent mesh
\[
h_{k,a}^2=\frac{4\delta}{m q_a(M,R)}
\]
whenever $q_a(M,R)>0$, and append the signed integer $K_{k,a}=\lfloor D_{k,a}/h_{k,a}\rceil$ to Alice's last message. Branches of zero posterior probability are omitted. If the realized answer is $B=a$, Bob reconstructs $Z_{k,a}$ by $b_{k,a}+h_{k,a}K_{k,a}$. Since $B$ has conditional distribution $q(\cdot\mid M,R)$, the average squared reconstruction error is at most $\delta$.

Let
\[
G_k=\E D_{k,B}^2.
\]
By \cref{lem:invariant}, all $D_{k,a}$ are determined before the final draw of $B$ and the original posterior $q_a(M,R)$ remains unchanged. Thus branch domination applies even after conditioning on the earlier auxiliary messages. Hence
\[
G_k\ge \rho\,\E\sum_a q_a(M,R)D_{k,a}^2,
\]
and therefore
\begin{equation}
\sum_a \E\frac{D_{k,a}^2}{h_{k,a}^2}
\le
\frac{m}{4\rho\delta}G_k.
\label{eq:branch-energy}
\end{equation}

\subsection{The innovation-energy telescope}\label{sec:energy}

Let $V_k$ denote the observer variance of $S_k$ in the current parent protocol. Let $v_k$ be the Bob variance of the $k$th centered observable in the coordinate child, and let $V_{k-1}$ be the observer variance of the residual sum after the new quantized innovations have been appended. Write
\[
\mathcal F_k=\sigma(F_k,B),\qquad
\mathcal G_k=\mathcal F_k\vee\sigma(X_{<k}),\qquad
U_k=\E\Var(S_{k-1}\mid\mathcal G_k).
\]
The rectangle property, including the fact that $B$ is Bob-side postprocessing, makes $S_{k-1}$ and $g_k(X_k,Y_k)$ conditionally independent given $\mathcal G_k$. Furthermore, conditioning on $Y_{<k}$ does not change the conditional law of $g_k$ given $\mathcal G_k$, so $v_k=\E\Var(g_k\mid\mathcal G_k)$. The law of total variance therefore yields
\begin{equation}
V_k\ge U_k+v_k+P_k,\quad
P_k=\E\left(\E[S_k\mid\mathcal G_k]-\E[S_k\mid\mathcal F_k]\right)^2.
\label{eq:pythagoras}
\end{equation}
The exact conditional mean of the residual sum given $\mathcal G_k$ is $Z_{k,B}$, and the quantized reconstruction has mean-square error at most $\delta$. The conditional-mean optimality principle gives $V_{k-1}\le U_k+\delta$. Consequently,
\begin{equation}
v_k+V_{k-1}\le V_k+\delta,
\qquad P_k\le V_k-v_k-V_{k-1}+\delta.
\label{eq:variance-step}
\end{equation}

The next estimate is where centering becomes useful. Put $H_k=g_k(X_k,Y_k)-\E g_k$ and
\[
L_k=\E[H_k\mid\mathcal G_k],\qquad
\overline L_k=\E[H_k\mid\mathcal F_k].
\]
Expanding $S_k=S_{k-1}+H_k$ gives
\[
D_{k,B}
=\Bigl(\E[S_k\mid\mathcal G_k]-\E[S_k\mid\mathcal F_k]\Bigr)
-(L_k-\overline L_k).
\]
By \cref{eq:variance-step}, the expected square of the first parenthesis is at most $V_k-v_k-V_{k-1}+\delta$. Conditional expectation is an $L_2$ contraction, so the expected square of the second parenthesis is at most $\Var(g_k(X_k,Y_k))$. Consequently,
\begin{equation}
G_k
\le
2\bigl(V_k-v_k-V_{k-1}+\delta\bigr)
+2\Var(g_k(X_k,Y_k)).
\label{eq:energy}
\end{equation}

Summing \cref{eq:energy} over the peeling process makes the variance terms telescope. Since $V_n\le\Var(S_n)=n\sigma^2$,
\begin{equation}
\sum_{k=2}^n G_k\le 4n\sigma^2+2n\delta.
\label{eq:total-energy}
\end{equation}

\subsection{From energy to information}

Before the quantized vector $(K_{k,a})_a$ is appended, the information costs of the coordinate child and the residual child sum to at most that of the parent. The new vector is a function of Alice's view. It reveals no additional information about Bob's residual input to Alice, while the information it reveals about Alice's residual input to Bob is at most its conditional entropy. By subadditivity and \cref{lem:entropy}, the information increase at stage $k$ is at most
\[
c_0m+\frac12\sum_{a:q_a>0}
\log_2\!\left(1+c_1\E\frac{D_{k,a}^2}{h_{k,a}^2}\right).
\]
There are at most $nm$ logarithmic terms. Jensen's inequality, followed by \cref{eq:branch-energy} and \cref{eq:total-energy}, bounds their total by
\[
Cnm\log_2\!\left(2+\frac{\sigma^2}{\rho\delta}\right).
\]
Thus the extracted protocols satisfy
\begin{equation}
\sum_{i=1}^n\IC(\eta_i)
\le
\IC(\pi)+Cnm\log_2\!\left(2+\frac{\sigma^2}{\rho\delta}\right).
\label{eq:info-final}
\end{equation}

On the variance side, summing \cref{eq:variance-step} and using \cref{eq:smoothing-variance} gives
\begin{equation}
\sum_{i=1}^n\Phi_{\rm B}(\eta_i;g_i)
\le
\Phi_{\rm obs}\!\left(\pi;\sum_i g_i\right)
+(2\rho-\rho^2)n\sigma^2+n\delta.
\label{eq:variance-final}
\end{equation}
Choose $\delta=\Delta/2$ and $\rho=\Delta/(4\sigma^2)$. The last two terms in \cref{eq:variance-final} sum to at most $n\Delta$. Moreover, $\sigma^2/(\rho\delta)=8(\sigma^2/\Delta)^2$, and hence $\log_2(2+\sigma^2/(\rho\delta))=O(\log_2(2+\sigma^2/\Delta))$. Dividing by $n$ proves \cref{thm:main}.

\section{Consequences}\label{sec:consequences}

\subsection{Gap-Majority}
Sawettamalya and Yu show that every protocol computing $\GapMAJ_n\circ f^n$ with success probability $0.99$ satisfies
\[
\Phi_{\rm obs}\!\left(\pi;\sum_{i=1}^n f(X_i,Y_i)\right)\le 0.99n
\]
when $f$ is balanced~\cite{SY26}. Applying \cref{thm:main} with $m=2$ and $\Delta=0.008$, then choosing a uniform extracted coordinate, gives Bob variance at most $0.998$ and information cost $\IC(\pi)/n+O(1)$. For a binary variable, posterior minority mass $p\le1/2$ satisfies $4p(1-p)\ge 2p$. Bob's posterior-majority decision consequently errs with probability at most $0.499$, proving \cref{cor:gapmaj}.

As $\IC(\pi)\le\mathrm{CC}(\pi)$, the same theorem yields a communication \emph{lower bound}. If $\IC_r^\mu(f;\alpha)$ is the minimum information cost of a single-instance protocol with success probability $\alpha$ under the same message convention, then
\[
\mathrm{CC}_r^{\mu^n}(\GapMAJ_n\circ f^n;0.99)
\ge n\bigl(\IC_r^\mu(f;0.501)-C\bigr).
\]
This does not bound the worst-case length of the extracted single-instance protocol.

\section{Terminal likelihood rank}\label{sec:rank}
The finite-answer theorem charges for every counterfactual terminal answer. A structural condition on the original output channel allows a different reduction: the terminal alphabet can be compressed before any coordinate is peeled. The compression need only preserve the squared conditional mean of the observable, rather than simulate the original answer.

Let $C=(M,R)$ be the common transcript and public randomness. Write
\[
q_c(da)=\Pr(A\in da\mid C=c),\qquad
K_c(y,da)=\Pr(A\in da\mid C=c,Y=y),
\]
Assume that $K_c(y,\cdot)\ll q_c$ for almost every $(c,y)$ and let $L_c(y,a)=dK_c(y,\cdot)/dq_c(a)$. The \emph{likelihood rank} at $c$ is the dimension of the linear span of the functions $L_c(y,\cdot)$ in $L_2(q_c)$, over almost every $y$. The construction below uses only the original terminal kernel. In the standard Borel setting, the finite-dimensional choices can be implemented measurably with fresh public randomness.

\begin{theorem}\label{thm:rank-reduction}
Suppose the original terminal channel has likelihood rank at most $s$ for almost every $c$. For every square-integrable observable $S=S(X,Y)$ there is a protocol with the same messages and public randomness, apart from additional independent public coins, whose terminal answer $J$ takes at most $s$ values and satisfies
\[
\E\Var(S\mid C,J,R')\le \E\Var(S\mid C,A).
\]
Here $R'$ denotes the added public randomness; Bob produces $J$ locally after the last message. In particular, the transcript information cost and the number of messages do not increase.
\end{theorem}

\begin{proof}
Fix $C=c$ and suppress $c$ from the notation. Denote the conditional distribution of $Y$ by $p$, and let $V\subset L_2(q)$ be the likelihood space, of dimension $d\le s$. Since $A$ is produced from Bob's view, $S$ and $A$ are conditionally independent given $(C,Y)$. Put
\[
 u(y)=\E[S\mid C=c,Y=y],\qquad
 f(a)=\E[S\mid C=c,A=a].
\]
For $h\perp V$ in $L_2(q)$, we have $\E[h(A)\mid C=c,Y]=0$, so
\[
 \int h(a)q(da)=0,\qquad
 \int f(a)h(a)q(da)=\E[u(Y)h(A)\mid C=c]=0.
\]
Thus both $1$ and $f$ belong to $V$. Fix a basis $\varphi_1=1,\varphi_2,\ldots,\varphi_d$ of $V$ and consider the integrable functions
\[
 \varphi_1,\ldots,\varphi_d,\ f^2.
\]
The finite-dimensional positive cubature theorem~\cite{BT06} yields points $a_1,\ldots,a_N$ and positive weights $w_1,\ldots,w_N$, with $N\le d+1$, such that
\begin{equation}\label{eq:cubature}
 \sum_{j=1}^{N}w_j\varphi_i(a_j)=\int\varphi_i\,dq
 \quad(1\le i\le d),\qquad
 \sum_{j=1}^{N}w_j f(a_j)^2=\int f^2\,dq.
\end{equation}
The points can be taken in a common $q$-full-measure set where all the identities below hold. Since $\varphi_1=1$, the weights sum to one.

Among nonnegative weights supported on these points satisfying the first $d$ equations in \eqref{eq:cubature}, maximize $\sum_jw_jf(a_j)^2$. The feasible set is a compact polytope, and an extreme optimum has at most $d$ positive weights. We retain these points and weights, relabeling them $1,\ldots,N$ with $N\le d$. The moment equations for $V$ remain exact, while
\begin{equation}\label{eq:cubature-gain}
\sum_{j=1}^{N}w_j f(a_j)^2\ge \int f(a)^2q(da).
\end{equation}

Bob now replaces his original answer by the distribution
\begin{equation}\label{eq:compressed-kernel}
\Pr(J=j\mid C=c,Y=y)=w_j L(y,a_j).
\end{equation}
It is nonnegative. Since $L(y,\cdot)\in V$, the moment equations in \eqref{eq:cubature} imply $\sum_jw_jL(y,a_j)=\int L(y,a)q(da)=1$ for $p$-almost every $y$. Thus \eqref{eq:compressed-kernel} is a probability kernel. Moreover,
\[
 \Pr(J=j\mid C=c)=w_j,\qquad
 \E[S\mid C=c,J=j]=f(a_j).
\]
The second identity follows from conditional independence and Bayes' rule. By \eqref{eq:cubature-gain},
\[
\E\left[\E[S\mid C=c,J]^2\mid C=c\right]
\ge \E\left[\E[S\mid C=c,A]^2\mid C=c\right].
\]
The conditional mean of $S$ given $C=c$ is unchanged, so the law of total variance proves the claim.

For standard Borel spaces, the construction can be implemented measurably with additional public coins. First obtain a basis of $V$ by starting from $1$ and applying Gram--Schmidt to a public i.i.d. sequence of likelihood functions $L(Y',\cdot)$, where $Y'\sim p$; finite rank ensures that these functions span $V$ almost surely. Next, conditional on $C=c$, use independent public coins to sample an infinite i.i.d. sequence from $q_c$. The moment vector of $(\varphi_1,\ldots,\varphi_d,f^2)$ belongs to the relative interior of the convex hull of its essential range, hence almost surely belongs to the convex hull of some finite prefix of the sampled vectors. Choose the first such prefix and a lexicographically selected feasible basic solution of the finite linear system \eqref{eq:cubature}; then perform the extreme-point optimization above. All selections involve finite-dimensional measurable linear programs. The fresh coins are independent of the original inputs and are known to both players, so the conditional argument and information cost are unchanged.
\end{proof}

\begin{corollary}\label{cor:rank-directsum}
Let the inputs and square-integrable observables satisfy the assumptions of \cref{thm:main}, except that the terminal answer can take values in a standard Borel space. If the original terminal channel has likelihood rank at most $s$ almost surely, then for every $0<\Delta\le\sigma^2$ there are $r$-message single-coordinate protocols such that
\[
\begin{aligned}
\frac1n\sum_i\IC_{\mu_i}(\eta_i)
&\le \frac{\IC(\pi)}n +Cs\log_2\!\left(2+\frac{\sigma^2}{\Delta}\right),\\
\frac1n\sum_i\Phi_{\rm B}(\eta_i;g_i)
&\le \frac1n\Phi_{\rm obs}\!\left(\pi;\sum_i g_i\right)+\Delta.
\end{aligned}
\]
\end{corollary}
\begin{proof}
Apply \cref{thm:rank-reduction} with $S=\sum_i g_i$ and then \cref{thm:main} to its at-most-$s$-symbol terminal answer. The additional public coins do not change the internal information cost.
\end{proof}

The rank bound concerns the original channel, not an intermediate protocol. For example, if $Y,A$ range over $\{0,\ldots,m-1\}$ and
\[
\Pr(A=a\mid Y=y)=\frac{1-\cos(2\pi(y-a)/m)}m,
\]
with $Y$ uniform, then the answer alphabet grows with $m$, but its likelihood rank is at most three. The conclusion of \cref{cor:rank-directsum} is independent of $m$.

\section{Arbitrary terminal outputs and mean-square estimation}\label{sec:mmse}
For terminal channels of unrestricted rank, a quantitative version is available if the sum has a finite fourth moment. Throughout this section the inputs and observables are i.i.d.

\begin{theorem}\label{thm:mmse}
Let $g(X,Y)$ be centered, with variance $\sigma^2>0$ and fourth moment finite. Write $S_n=\sum_{i=1}^n g(X_i,Y_i)$ and $\kappa_n=\E S_n^4/(n^2\sigma^4)$. If an $r$-message protocol with an arbitrary terminal answer satisfies
\[
\Phi_{\rm obs}(\pi;S_n)\le(1-\gamma)n\sigma^2,
\qquad 0<\gamma\le1,
\]
then an $r$-message single-copy protocol $\eta$ satisfies
\[
\begin{aligned}
\Phi_{\rm B}(\eta;g)&\le(1-c\gamma)\sigma^2,\\
\IC(\eta)&\le\frac{\IC(\pi)}n
 +C\log_2\!\left(2+\frac{\sqrt{\kappa_n}}\gamma\right)
 \log_2\!\left(2+\frac1\gamma\right),
\end{aligned}
\]
for absolute constants $c,C>0$.
\end{theorem}
\begin{proof}
Put $C_0=(M,R)$ and $U=\E[S_n\mid C_0,A]-\E[S_n\mid C_0]$. If $C_0$ alone explains at least $\gamma n\sigma^2/2$ of the variance, discard the terminal answer and apply \cref{thm:main} with a singleton answer and $\Delta$ a sufficiently small multiple of $\gamma\sigma^2$. Otherwise orthogonality gives
\[
\E U^2\ge\frac{\gamma n\sigma^2}2,
\qquad
\E U^4\le16\kappa_n n^2\sigma^4.
\]
Set $a=\sqrt{\gamma n\sigma^2}/16$ and $b=64\sqrt{\kappa_n/\gamma}\sqrt{n\sigma^2}$. Map $U$ to zero when $|U|<a$ or $|U|>b$. On the remaining range, retain the sign and round the magnitude down to a geometric grid whose ratio is at most $5/4$. The resulting discrete variable $Q$ satisfies
\[
\E(U-Q)^2\le\frac1{25}\E U^2+\frac{\gamma n\sigma^2}{128}.
\]
Indeed, the two discarded tails contribute at most $a^2$ and $\E U^4/b^2$, respectively. Since the conditional expectation $\E[U\mid C_0,Q]$ is the least-squares predictor of $U$, this quantization preserves a constant fraction of $\E U^2$. There are
\[
 m=O\left(\log\left(2+\frac{\sqrt{\kappa_n}}\gamma\right)\right)
\]
possible values of $Q$. Applying \cref{thm:main} to the protocol with terminal answer $Q$, with $\Delta$ a small multiple of $\gamma\sigma^2$, proves the theorem.
\end{proof}

For a fixed square-integrable observable with finite fourth moment,
\[
\kappa_n=3+\frac{\E g^4/\sigma^4-3}{n}.
\]
Consequently the additive information loss in \cref{thm:mmse} is $O_g(\log^2(2+1/\gamma))$. A single logarithm is sufficient under an additional condition on the posterior gain.

\begin{corollary}\label{cor:smallball}
In the setting of \cref{thm:mmse}, suppose $U$ satisfies $\E U^4\le K(\E U^2)^2$ whenever $\E U^2>0$, where $K\ge1$. Then the conclusion holds with
\[
\Phi_{\rm B}(\eta;g)\le(1-c\gamma/K)\sigma^2,
\qquad
\IC(\eta)\le\frac{\IC(\pi)}n+C\log_2(2+K/\gamma).
\]
\end{corollary}
\begin{proof}
If the transcript alone explains at least $\gamma n\sigma^2/2$, use a singleton answer. Otherwise $\E U^2\ge\gamma n\sigma^2/2$. Bob replaces the answer by $\operatorname{sgn}(U)$. Conditional projection and $L_1$--$L_2$--$L_4$ interpolation give
\[
\E\left[\E[U\mid C_0,\operatorname{sgn}(U)]^2\right]
\ge (\E|U|)^2
\ge \frac{(\E U^2)^3}{\E U^4}
\ge\frac{\gamma n\sigma^2}{2K}.
\]
Apply \cref{thm:main} with a binary answer and $\Delta$ a sufficiently small multiple of $\gamma\sigma^2/K$.
\end{proof}

\section*{AI assistance}
OpenAI ChatGPT was used during exploratory derivations, proof review, literature checking, and manuscript editing. The author is responsible for the arguments and statements in this manuscript.

\appendix

\section{One-step peeling and information identities}\label{app:peeling}
We record the information and variance identities used in the proof. They follow the coordinate decomposition of~\cite{SY26}, with the terminal auxiliary message changed as described in Section~\ref{sec:proof}.

Consider a parent protocol on $(X_1,Y_1),\ldots,(X_k,Y_k)$ drawn from $\prod_{i=1}^k\mu_i$. In the residual child, $(X_{<k},Y_{<k})$ are the real inputs. The players publicly sample $Y_k$ from its marginal, and Alice privately samples $X_k$ from $\mu_k(\cdot\mid Y_k)$. In the coordinate child, $(X_k,Y_k)$ are the real inputs. The players publicly sample $X_{<k}$ from its marginal, and Bob privately samples $Y_{<k}$ from the appropriate conditional product distribution. Both children can simulate the parent transcript $M$ and Bob's original terminal output, and both reproduce the parent's joint law before the extra message is added.

The internal information costs of the two simulated transcripts can be compared by conditioning additionally on the private randomness used to sample missing coordinates. This gives an upper bound because the missing coordinate is sampled independently of the real input of the other player, conditional on the public samples and that player's real input. For the Alice-input contributions, the two terms are bounded by
\[
I(M:X_{<k}\mid Y^k,R),\qquad
I(M:X_k\mid X_{<k},Y^k,R),
\]
respectively; for the Bob-input contributions, they are bounded by
\[
I(M:Y_{<k}\mid X^k,Y_k,R),\qquad
I(M:Y_k\mid X^k,R).
\]
In both cases the bounds sum, by the chain rule, to the corresponding information term of the parent. Thus the total information cost before adding the auxiliary message is at most $\IC(\pi)$. This argument applies to every residual protocol, since its public samples remain independent of the real child inputs.

Only the residual child receives the quantized vector $K=(K_{k,a})_a$, appended to Alice's final message. Since $K$ is computed from Alice's view, it reveals no additional information about Bob's real input conditional on Alice's input and private randomness. The additional Alice-input information is bounded by
\[
I(K:X_{<k}\mid M,Y_{<k},Y_k,R)
\le H(K\mid M,Y_k,R),
\]
which is the conditional entropy estimated in the main proof.

For the variance calculation, condition on the parent transcript and public randomness, the resampled answer $B$, the exposed $Y_k$, and $X_{<k}$. A private-coin two-party transcript has the rectangle property. More explicitly, after conditioning on $M=m$, $R=r$, and Bob-side output $B=a$, the density of the input pairs is proportional to
\[
\left(\prod_{i=1}^k\mu_i(x_i,y_i)\right)
\alpha_{m,r}(x_1,\ldots,x_k)\,\beta_{m,r,a}(y_1,\ldots,y_k).
\]
Fixing $X_{<k}$ and $Y_k$ separates the remaining factors involving $X_k$ from those involving $Y_{<k}$. Thus $X_k$ and $Y_{<k}$ are conditionally independent, and so are $g_k(X_k,Y_k)$ and the residual sum $S_{k-1}$. The same factorization shows that conditioning on $Y_{<k}$ does not change the conditional law of $X_k$. More generally, at each recursive stage the previously appended auxiliary data only changes the Alice-side factor $\alpha$, because each such vector is generated before the terminal output and from Alice's view and public samples. Bob's deferred terminal kernel only changes the Bob-side factor $\beta$. Thus the conditional independence remains valid for the complete residual transcript, not just for the first decomposition step.

Consequently, with the sigma-fields $\mathcal F_k,\mathcal G_k$ defined in the main proof, the law of total variance yields
\[
\begin{aligned}
\E\Var(S_k\mid\mathcal F_k)
&=\E\Var(S_k\mid\mathcal G_k)
 +\E\bigl(\E[S_k\mid\mathcal G_k]-\E[S_k\mid\mathcal F_k]\bigr)^2,\\
\E\Var(S_k\mid\mathcal G_k)
&=\E\Var(S_{k-1}\mid\mathcal G_k)+v_k.
\end{aligned}
\]
Since conditioning additionally on $Y_k$ cannot increase expected conditional variance, these identities give~\eqref{eq:pythagoras}. Finally, Bob's reconstructed conditional mean of $S_{k-1}$ has mean-square error at most $\delta$ relative to $\E[S_{k-1}\mid\mathcal G_k]$. The optimality of conditional expectation then implies $V_{k-1}\le U_k+\delta$, which is~\eqref{eq:variance-step}.

\end{document}